\documentclass[a4paper,fleqn]{cas-dc}
\usepackage[authoryear,longnamesfirst]{natbib}
\usepackage{graphicx}
\usepackage{mathtools}
\usepackage{enumitem}
\usepackage{xcolor}
\let\cite\citep
\newtheorem{theorem}{Theorem}

\newtheorem{cor}[theorem]{Corollary}
\newtheorem{prop}[theorem]{Proposition}
\newtheorem{remark}[theorem]{Remark}
\newtheorem{definition}[theorem]{Definition}
\newproof{proof}{Proof}

\newcounter{eggsample}[section]

\newenvironment{eggsample}[1][]{\refstepcounter{eggsample}\par\smallskip
   \paragraph{\thesection.\theeggsample.} \emph{#1}}%
{\hfill\smallskip}
\numberwithin{equation}{section}

\begin{document}
\let\WriteBookmarks\relax
\def\floatpagepagefraction{1}
\def\textpagefraction{.001}

\shorttitle{RAF sets and stoichiometric autocatalysis}    

\title[mode = title]{Bridging two theoretical frameworks of autocatalysis:
  \textsc{RAF} sets and stoichiometric autocatalysis}

\author[1,2]{Richard Golnik}

\author[1]{Thomas Gatter}

\author[1]{Wim Hordijk}

\author[1,2,3,4,5,6,7,8]{Peter F. Stadler}

\author[1]{Nicola Vassena}

\affiliation[1]{organization={Leipzig University,
    Bioinformatics Groups, Department of Computer Computer Science and
    Interdisciplinary Center for Bioinformatics},
  addressline={H{\"a}rtelstrasse 16-18}, 
  city={Leipzig},
  postcode={04107}, 
  country={Germany}
}

\affiliation[2]{organization={Leipzig University,
    Zuse School of Embedded and Composite AI (SECAI)},
  addressline={H{\"a}rtelstrasse 16-18}, 
  postcode={04107}, 
  city={Leipzig},
%
  country={Germany}
}

\affiliation[3]{organization={Leipzig University, Center for Scalable Data
    Analytics and Artificial Intelligence, ScaDS.AI Dresden/Leipzig},
  addressline={L{\"o}hrs Carr{\'e}, Humboldtstraße 25, Uferstr. 11},
  city={Leipzig},
  postcode={04105}, 
  country={Germany}
}

\affiliation[4]{organization={Leipzig University, LeiCeM --
    Leipzig Center of Metabolism},
  addressline={},
  city={Leipzig},
  postcode={04105}, 
  country={Germany}
}

\affiliation[5]{organization={Max Planck Institute for Mathematics in the Sciences},
  addressline ={Inselstra{\ss}e 22},
    postcode={04103}, city=Leipzig, country=Germany}

\affiliation[6]{
  organization={University of Vienna, Institute of Theoretical Chemistry},
  addressline={W{\"a}hringerstra{\ss}e 17},
  postcode={A-1090},
  city=Wien,
  country=Austria}

\affiliation[7]{organization={Facultad de Ciencias, Universidad National de
    Colombia},
  city={Bogot{\'a}},
  country={Colombia}
}

\affiliation[8]{organization={Santa Fe Institute}, addressline={1399 Hyde
    Park Rd.}, city=Santa Fe, state=NM, postcode=87501,
  country=USA}

\begin{abstract}
  Autocatalysis lies at the heart of many (bio)chemical processes and
    is key to processes leading up to the origin of life. Two seemingly very
    different formalisms have emerged that define autocatalysis. Kauffman introduced \emph{collective autocatalysis} to
    describe systems of molecules that mutually catalyze each other's
    formation, emphasizing the self-sustaining character of autocatalytic
    systems. This view is mathematically formalized in the theory of
    \emph{Reflexively Autocatalytic and Food-generated sets} (RAF). In
    parallel, \emph{stoichiometric autocatalysis} emerged from the theory of
    Chemical Reaction Networks (CRN), focusing on the net-productive,
    self-amplifying character of autocatalytic subnetworks. These two
    frameworks have coexisted independently in the literature, since RAF
    theory considers each reaction as explicitly catalyzed, while the CRN
    approach often excludes explicitly catalyzed
    reactions altogether. Nevertheless, both frameworks describe reaction networks
    and thus admit a common mathematical representation in terms of
    stoichiometric matrices. We highlight this connection and show that the
    two formalisms are less disparate than they might appear. To illustrate
    this point we prove that, under mild and general conditions, any RAF is stoichiometrically autocatalytic.
\end{abstract}

\maketitle

\section{Introduction}

Autocatalysis refers to a process in which complex systems sustain and
reinforce their own existence. In its simplest form, autocatalysis takes
the shape of a single autocatalytic reaction; that is, a chemical
reaction in which one of the species catalyzes its own formation as
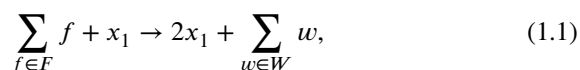
\begin{equation}\label{eq:autocatreaction}
	\sum_{f\in F} f + x_1 \rightarrow 2x_1 + \sum_{w\in W} w,
\end{equation}
under the assumption of a sufficient supply of food molecules $F$ and the
potential release of waste molecules $W$. This idea dates back to the end
of the 19th century, when Wilhelm Ostwald proposed the concept
\cite{Ostwald:90}.

Transformations of the form \eqref{eq:autocatreaction} are not
restricted to chemistry but emerge in different contexts, such as systems
of predator-prey models, bacterial growth, or epidemiology. A convincing
example is for instance the disease transmission from an infected
individual to a susceptible one, which results in a total of two infected
individuals, but also the mandatory feeding of predators on their prey or
of bacteria, such as \emph{E.\ coli}, on glucose to reproduce, yields the
same kind of reactions. Networks of such self-replicating
entitities---\emph{replicators}---appear already in Eigen's quasispecies
\cite{eigen_selforganization_1971} and most prominently in the hypercycle
model that describes autocatalytic self-replication dependent on an
additional explicit catalyst \cite{eigen_principle_1977}.

Generalizing the idea of autocatalysis from single autocatalytic
reactions to subsystems capable of sustaining their own reproduction has
given rise to multiple notions of autocatalysis
\cite{andersen_defining_2021}, and in some cases has led to confusion in
the usage of distinct concepts of autocatalysis
\cite{hordijk17confusion}. A natural extension of autocatalytic reactions
are autocatalytic cycles \emph{sensu} Barenholz
\cite{barenholz_design_2017}. They are sets of reactions that form a loop,
which yields a net stoichiometric increase of one of its intermediate
metabolites. Classic examples of them can be found at the heart of carbon
metabolism, such as the degradation of glucose to two molecules of
phosphoenolpyruvate, whose conversion to pyruvate is, in turn, required for
glucose uptake or the reverse Tricarboxylic Acid (TCA) cycle
\cite{smith_universality_2004}. Blokhuis et al. developed the idea into the
framework of \emph{stoichiometric autocatalysis} in chemical reaction
networks free of explicit catalysis \cite{blokhuis_universal_2020},
starting from the IUPAC definition of an autocatalytic reaction
\cite{gold_iupac_2025}. In contrast, \textsc{autocatalytic sets}
\emph{sensu} Kauffmann \cite{kauffman1986} take a different approach as they are
defined in systems of reactions where each reaction is catalyzed. Both of
these concepts are closely connected with ideas on how life on Earth may
have developed \cite{kauffman1986,
  unterberger_stoechiometric_2022}. However, due to their motivation, they
are inherently distinct in their mathematical formalization.

Kauffman \cite{kauffman1986} aimed to develop a theoretical
framework that supported the existence of reflexively autocatalytic sets of
protein-like complexes, i.e., sets of polypeptides in which each reaction,
leading to the formation of one of its members, is catalyzed by at least
one member of the set. Assuming a fixed probability for each
  polypeptide to catalyze the formation of other polypeptides within the
  set, a maximum polymer size, and sufficient availability of the
  monomeric building blocks, Kauffman showed that there is a high
  probability for the emergence of self-sustaining prebiotic biochemistry.
The concept has been rigorously formalized as \emph{Reflexively
Autocatalytic and Food-generated sets} (RAFs) \cite{steel_emergence_2000,
  hordijk04}, and subsequently refined and studied in various ways
\cite{mossel05CAF, steel13minimal, smith_autocatalytic_2014, hordijk2019,
  steel20}.

One biochemically relevant restriction of the original RAF framework is
CAFs, \emph{Constructively Autocatalytic Sets} (CAFs)
\cite{mossel05CAF}. While in a RAF, at least initially, reactions may also
take place in an uncatalyzed manner, CAFs require that a catalyst of a
reaction and all of its reactants must either have been generated
previously or belong to the food set. This changes the biological
perspective, since reactions are then more likely to be catalyzed by
species that are located upstream of them, but not by their direct or
downstream products. As a consequence, a larger food set may be required, in
a ``Kauffman setting'', smaller polypeptides would need to catalyze the
generation of larger polypeptides, which is counter-intuitive to the
original motivation that assigned a catalytic activity to almost completely
synthesized proteins. The CAF framework thereby seems to be rather
restrictive \cite{mossel05CAF}, in particular, for origin-of-life
biological interpretations.

In contrast, stoichiometric autocatalysis
\cite{barenholz_design_2017,blokhuis_universal_2020} emerged from the
theory of chemical reaction networks and is aligned with the IUPAC
definition of (auto-)catalysis \cite{gold_iupac_2025} and formalized in a
manner that differs substantially from the RAF framework. While RAFs
require each reaction to be catalyzed by a member of the set of involved
species, catalysis in CRNs is often seen as the combined effect of a set of
reactions that can occur in sequence or in parallel; however, a single
reaction may or may not be explicitly catalyzed. In fact, most accounts
exclude explicit catalysis and consider reactions as reversible
\cite{blokhuis_universal_2020, kosc_thermodynamic_2025}. This theory has
only recently been extended to general CRNs composed of both reversible and
irreversible reactions that may, but need not, involve explicit catalysis
\cite{golnik_autocatalytic_2026}.

Moreover, \textsc{RAF}s, viewed as subnetworks, always include `complete'
reactions with all their reactants and products. Stoichiometrically
autocatalytic subnetworks \cite{blokhuis_universal_2020}, on the other
hand, may contain only a subset of the involved species so that each
reaction in the subnetwork is \emph{well-formed}, i.e., it has at least one
reactant and one product in such subset.  The net increase of a species in
an explicit autocatalytic reaction translates into a net increase in the
concentration of each species of a stoichiometric autocatalytic subnetwork,
in which all reactions exhibit a strictly positive rate
\cite{blokhuis_universal_2020}. Accordingly, the corresponding submatrices
are semipositive. Together, well-formedness and semipositivity guarantee
that each species is produced by at least one reaction. Species-autonomy,
i.e., each species also being a reactant for at least one reaction, can be
derived as a consequence from other subnetwork properties, such as
minimality or irreducibility \cite{vassena_unstable_2024,
  blokhuis_universal_2020}. The matrix-based formalization of
stoichiometric autocatalysis recently led to the development of efficient
algorithms enumerating autocatalytic subsystems
\cite{gagrani_polyhedral_2024, kosc_thermodynamic_2025, golnik_birne_2025,
  golnik_enumeration_2026, golnik_using_2026}.  A variant of stoichiometric
autocatalysis that emphasizes the notion of seed-dependence, conceptually
similar to the food set in RAFs, are the seed-dependent autocatalytic
system (SDAS) studied in \cite{Peng:22}.

Despite conceptual similarities, there is no simple correspondence between
the two frameworks and -- as discussed in \cite{Peng:22}, there are systems
that are deemed autocatalytic in one theory but not in the other. The
semipositivity of autocatalytic submatrices heavily depends on the
stoichiometric coefficients, while the \textsc{RAF} framework is {a priori}
ignorant of stoichiometric requirements. From a formal point of view, the
frameworks thus appear to be completely different at first
glance. Nevertheless, meaningful connections can be made, as
  demonstrated by the relationships with Chemical Organization Theory
  \cite{dittrich2007organization} that has led to a notion of
  \textsc{closed RAF} in \cite{Hordijk:18}. As a first step towards
exploring the connection with CRN Theory, we show here that \textbf{every
  RAF that is not itself a CAF contains a stoichiometrically autocatalytic
  subnetwork}.

\section{Reaction Networks}\label{sec:preliminaries}

\paragraph{Basic concepts and notation.}
A reaction network $\Gamma=(X,R)$ is a pair of sets, where
$X=(x_1,...,x_{|X|})$ is the set of species and $R=(r_1,...,r_{|R|})$ is
the set of reactions. Each reaction $r_j$ is an ordered association of
nonnegative combinations of the species:
\begin{equation}\label{eq:reactionj}
 r_j\sum_{m=1}^{|X|} s^-_{mj} \; x_m \quad \underset{r_j}\longrightarrow
 \quad \sum_{m=1}^{|X|} s^+_{mj}\;x_m,
\end{equation}
where the nonnegative $s^-_{mj}$ and $s^+_{mj}$ are the
\emph{stoichiometric coefficients}. For a reaction $r_j$, a species $x_m$
is called a \emph{reactant} if $s^-_{mj}>0$, a \emph{product} if
$s^+_{mj}>0$, and a \emph{catalyst} if it is both a reactant and a product,
i.e. if $s^-_{mj}s^+_{mj}>0$. The sets of reactants, products, and
catalysts of $r_j$ are denoted $\rho(r_j)$, $\pi(r_j)$, and $\kappa(r_j)$,
respectively. The union set $\rho(r_j)\cup \pi(r_j)$ of reactants and
products is called \emph{support} of $r_j$ and denoted by $X(r_j)$, and,
for a set of reactions $R$, $X(R)\coloneqq \bigcup_{r\in R}
X(r)$. Moreover, a reactant $x_m$ of a reaction $r_j$ is called
\emph{net-reactant} if $s^-_{mj}>s^+_{mj}\ge0$. Conversely, a product $x_m$
of reaction $r_j$ is called \emph{net-product} if
$s^+_{mj}>s^-_{mj}\ge0$. The sets of net-reactants and net-products of
$r_j$ are indicated respectively as $\rho_{net}(r_j)$ and
$\pi_{net}(r_j)$. The stoichiometric matrix $S$, encoding only the
net-stoichiometry of the reactions, is then defined as follows:
\[\label{eq:stoichiomatrix}
S_{mj}\coloneqq s^{+}_{mj}-s^-_{mj}.
\]
In particular, negative entries of $S$ identify net-reactants and positive entries of $S$ identify net-products. 

\paragraph{Subnetworks and well-formed subnetworks.} A subnetwork
$\Gamma' \subseteq \Gamma$ is any network $\Gamma'=(X',R')$ with
$X'\subseteq X$ and $R'\subseteq R$, where the stoichiometry of any
$r_j \in R'$ is naturally restricted to $X'$ as $s^{-}_{mj}$ (or
$s^{+}_{mj}$) if $x_m \in X'$, and zero otherwise. That is, we omit any
reference to possible participating species $x_m\not\in X'$ in
$r_j$. Consequently, the stoichiometric matrix $S'$ of a subnetwork
$\Gamma'=(X',R')$ corresponds to the submatrix of $S$ identified by
restricting to species-rows corresponding to $X'$ and reaction-columns
corresponding to $R'$, in symbols we write
\[S'=S[X',R'].\]

\begin{definition}\label{def:wellformed}
A subnetwork $\Gamma'=(X',R')$ $\subseteq \Gamma$ is called \emph{well-formed} if for all reactions $r_j\in R'$ there exist two species $x_m, x_n \in X'$, not necessarily distinct, such that $s^-_{mj}s^+_{nj}>0$. Consistently, a submatrix $S[X',R']$ of the stoichiometric matrix is said to be \emph{well-formed} if its associated subnetwork is well-formed.
\end{definition}
Said otherwise, if $\Gamma'=\Gamma$, well-formed networks are the ones that do not include inflow reactions, with no reactants, and outflow reactions, with no products. However, when the concept is applied to strict subnetworks $\Gamma'=(X',R')$ with $X'\subsetneq X$, Def.~\ref{def:wellformed} requires the presence, for each reaction $r_j\in R'$, of at least one reactant and one product in the subset $X'$. Note also that if a species $x_m$ is a catalyst of $r_j$ but not a net-reactant nor a net-product, i.e. $s^+_{mj}=s^-_{mj}\neq 0$, then the associated entry of the stoichiometric matrix is zero: $S_{mj}=s^+_{mj}-s^-_{mj}=0$. Nevertheless, the subnetwork $(\{x_m\},\{r_j\})$ is well-formed, and so it is its stoichiometric matrix $S[\{x_m\},\{r_j\}]=0$.

\paragraph{Stoichiometric Autocatalysis.} 

The definition of autocatalysis by the International Union of Pure and Applied Chemistry \cite{IUPAC+C00876+2025} reads: \begin{center}\emph{
Catalysis brought about by one of the products of a reaction is called \emph{autocatalysis}.}\end{center}
Building on this definition and under an emergent-property perspective, Blokhuis et al. \cite{blokhuis_universal_2020} derived a general stoichiometric characterization of autocatalysis. While their original formulation excluded catalysts, later work \cite{vassena_unstable_2024, golnik_autocatalytic_2026} lifted such definition to a framework that allows catalysts. We follow here such general framework. We call a vector $v\in\mathbb{R}^n$ \emph{positive}, and we write $v>0$, if all of its entries are positive, i.e. $v_j>0$ for $j=1,...,n.$.
\begin{definition}\label{def:semipositive}
A matrix $S$ is called \emph{semipositive} if there exists a positive vector $v>0$ such that $Sv>0$.
\end{definition}
Then, the combination of Def.~\ref{def:wellformed} and Def.~\ref{def:semipositive} just brings the following.

\begin{definition}[Stoichiometric autocatalysis]\label{def:autocat}
A stoichiometric matrix is \emph{autocatalytic} if it is well-formed and semipositive. A reaction network $\Gamma$, whose stoichiometric matrix  $S $ possesses an autocatalytic submatrix $S'$ is called \emph{stoichiometrically autocatalytic}.
\end{definition}
We remark that if catalysts are absent, and thus each reactant is a net-reactant and each product is a net-product, then the stoichiometric matrix $S$ itself encodes the full network since its negative entries identify all reactants and its positive entries identify all products. In this context, then, the definition of stoichiometric autocatalysis is really only a matrix definition. In the present more general setting, in turn, an inspection of the network stoichiometry fully encoded by $R$, with its catalysts, and not only of its net-stoichiometry, encoded by the stoichiometric matrix $S$ alone, is needed.

\section{Reflexively Autocatalytic Food-generated Sets}

In the context of origin of life and collectively autocatalytic systems formalized in the \textsc{RAF} framework, some of the catalysts may not be present initially. For example, a reaction that has all of its reactants in the food set but which is catalyzed by one of its own products, cannot happen catalyzed when only food molecules are present. However, since catalysts primarily speed up the rate at which reactions happen, such a reaction can still occur spontaneously (i.e., uncatalyzed, at a lower rate). Once it does, the first catalyst has been produced, and the reaction can now proceed in an (auto)catalytic manner. In fact, allowing some catalysts to be produced only at a later stage is exactly what enables autocatalytic sets to be \emph{evolvable}, at least to some extent \cite{vasas2012,hordijk2018}.

For this reason, for each reaction $r_j$ we equivalently rewrite \eqref{eq:reactionj} as: 
{\small \begin{equation}\label{eq:Freactionj}
        \sum_{m=1}^{|X|} \sigma^-_{mj} \; x_m  \;+\; \bigg(\sum_{m=1}^{|X|} k_{mj}x_m \bigg)\;\underset{r_j}\longrightarrow \; \sum_{m=1}^{|X|} \sigma^+_{mj}\;x_m\; + \;\bigg(\sum_{m=1}^{|X|} k_{mj}x_m \bigg),        
\end{equation}}
with the following constraints that guarantee that $r_j$ is the same in \eqref{eq:reactionj} and \eqref{eq:Freactionj}:
\begin{equation}\label{eq:constraintsF}
  \sigma^-_{mj} \; + \; k_{mj}= s^-_{mj} \qquad
  \sigma^+_{mj} \; + \; k_{mj}= s^+_{mj}.
\end{equation}
Species $x_m$ appearing in \eqref{eq:Freactionj} with $\sigma^-_{mj}>0$ (resp. $\sigma^+_{mj}>0$)  are called \textsc{F-reactants} (resp. \textsc{F-products}) of $r_j$, and the set is denoted as $\rho_{\text{F}}(r_j)$ (resp. $\pi_{\text{F}}(r_j)$). Species $x_m$ appearing in \eqref{eq:Freactionj} with $k_{mj}>0$ are called \textsc{F-catalysts}, with set denoted as $\kappa_{\text{F}}(r_j)$. 

\begin{remark}[Original RAF framework I]\label{rmk:eqclas}
 In the original RAF framework, one single reaction is specified only by its \textsc{F-reactants} and \textsc{F-products}, and it can be catalyzed, independently, by multiple choices of \textsc{F-catalysts}. To stay closer to a general reaction network framework, here we consider one reaction catalyzed by different \textsc{F-catalysts} as different reactions, as their formulation \eqref{eq:Freactionj} differs.
\end{remark}

We further stress that the \textsc{RAF} framework does not require any pairwise empty intersection among its defining sets. In particular, a species $x_m$ might be at once an \textsc{F-reactant}, an \textsc{F-product}, and a \textsc{F-catalyst} of one single reaction $r_j$, see examples in \ref{sect:examples}.\ref{egg:ex2}. Besides the obvious relation that the set of reactants (resp. products) is the union of the sets of \textsc{F-reactants} (resp. \textsc{F-products}) and \textsc{F-catalysts}, i.e.,
\[
    \rho(r)= \rho_{\text{F}}(r) \cup \kappa_{\text{F}}(r)  \quad \text{and} \quad  \pi(r)=\pi_{\text{F}}(r) \cup \kappa_{\text{F}}(r),
\]
the only set relations that can be assured, are stated in the following Proposition.
\begin{prop}\label{prop:inclusions}
Let $r$ a reaction written as \eqref{eq:Freactionj}. The following holds.
\begin{enumerate}
    \item Any net-reactant is an \textsc{F-reactant}, and any \textsc{F-reactant} is a reactant:
    \[
        \rho_{\text{net}}(r) \subseteq \rho_{\text{F}}(r) \subseteq \rho(r)\]
    \item  Any net-product is an \textsc{F-product}, and any \textsc{F-product} is a product:
    \[\pi_{\text{net}}(r) \subseteq \pi_{\text{F}}(r) \subseteq \pi(r)\]
    \item Any \textsc{F-catalyst} is a catalyst, and in particular both a reactant and a product:
    \[
    \kappa_{\text{F}}(r) \subseteq \kappa(r) =\rho(r) \cap \pi(r).
    \] 
\end{enumerate}
\end{prop}
\begin{proof}
The implications $ \rho_{\text{F}}(r) \subseteq \rho(r)$, $ \pi_{\text{F}}(r) \subseteq \pi(r)$, and $\kappa_{\text{F}}(r) \subseteq \kappa(r)=\rho(r) \cap \pi(r)$ follow directly from the definition. To prove the implication $\rho_{net}(r)\subseteq\rho_{\text{F}}(r)$, let us take a net-reactant $x_m$ of a reaction $r_j$, i.e., $x_m\in \rho_{net}(r_j)$. By \eqref{eq:constraintsF} and definition of net-reactant, we have that
    $\sigma^-_{mj}+k_{mj}=s^-_{mj}>s^+_{mj}=\sigma^+_{mj}+k_{mj}$,
which yields
$\sigma^-_{mj}>\sigma^+_{mj}\ge 0$,
and thus $x_m$ is an \textsc{F-reactant} of the reaction $r_j$. The proof of the implication $\pi_{net}(r)\subseteq \pi_{\text{F}}(f)$ is identical as the one for the reactants, just based on the opposite inequality and the definition of net-product.
\end{proof}

An assumption that appears natural and relevant in realistic modeling of
biochemical networks could be that the net-reactants (net-products) coincide with the
\textsc{F-reactants} (\textsc{F-products}) and  \textsc{F-catalysts}
coincide with catalysts. Since this is not assumed in the classic RAF
framework, we proceed here in full generality.


\begin{definition}\label{def:RAF}
A \textsc{Reflexively Autocatalytic Food-generated set} (\textsc{RAF}) is a
reaction network $\Gamma_F=(X,R,F)$ with a specified \emph{food set}
$F\subseteq X$ such that the following properties hold:
\begin{enumerate}[label={R\arabic*.}]
\item For each reaction $r\in R$, the set of \textsc{F-catalysts} is nonempty:
\[
    \kappa_{\text{F}}(r)\neq\emptyset,\quad \text{for all $r\in R$.}
\]
    \item Any \textsc{F-catalyst} of a reaction $r_j \in R$ is either in the food set or it is an \textsc{F-product}, but not an \textsc{F-reactant}, of at least one reaction in $R$, in symbols:
    \begin{equation}\label{eq:rafc2}
       \kappa_{\text{F}}(r_j) \; \setminus \; F \quad\subseteq \quad \bigcup_{r\in R} \quad \bigg(\pi_{\text{F}}(r)\;\setminus\; \rho_{\text{F}}(r)\bigg) 
    \end{equation}
    \item There exists a partial order $\preceq$ on $R$ such that all \textsc{F-reactants} of a reaction $r_j\in R$ either are in the food set or they are \textsc{F-products}, but not \textsc{F-reactants}, of at least one reaction $r_{j'}$ with $r_{j'} \prec r_{j}$. In symbols,
\begin{equation}\label{eq:rafc3}
\rho_{\text{F}}(r_j)\;\setminus\; F \quad \subseteq \quad \bigcup_{r \;|\;r\;\prec\;r_j} \bigg(\pi_{\text{F}}(r) \;\setminus \; \rho_{\text{F}}(r)\bigg).
\end{equation}
\end{enumerate}
\end{definition}

Properties R1 and R2, jointly, define the network to be
\textsc{Reflexively Autocatalytic} (RA) while property R3 defines it to be
\textsc{Food-generated} (F). We refer to the partial order $\preceq$ as the
\textsc{generation order} and, by convention, we represent it via a
\textsc{generation map} $\gamma$ that takes values in an interval
$I\coloneqq [1,...,N] \subset \mathbb{N}$, i.e.  $\gamma: R \mapsto
\mathbb{N}_{>0}$ with $r_{j'}\preceq r_{j}$ if $\gamma(r_{j'}) \le
\gamma(r_j)$. More specifically, a reaction $r$ is in the preimage of 1,
i.e. $r \in \gamma^{-1}(1)$, if its \textsc{F-reactants} are in the food
set, and we say that $r$ is of \textsc{first generation}. Iteratively, a
reaction is in the preimage of $h\in I$, i.e. $r_j \in \gamma^{-1}(h)$, if
it is not in the preimage of $h'<h$, and its \textsc{F-reactants} are
either in the food set or they are \textsc{F-products}, but not
\textsc{F-reactants}, of at least one reaction $r_{j'}$ with $r_{j'} \in
\gamma^{-1}(h')$, $h'<h$, and we say that $r_j$ is of $h^{\text{th}}$
generation. In order to better connect to stoichiometric properties, we
restate properties R2 and R3 of Def.~\ref{def:RAF} directly in terms of
net-products: 
\begin{prop}\label{prop:equivnet}
Conditions \eqref{eq:rafc2} and \eqref{eq:rafc3} in Def.~\ref{def:CAF} can equivalently be re-stated in terms of net-products as follows:
\begin{equation}\tag{\ref{eq:rafc2}$'$}\label{eq:rafc2'}
 \kappa_{\text{F}}(r_j) \; \setminus \; F \quad\subseteq \quad \bigcup_{r\in R} \quad \pi_{net}(r),
\end{equation}
and, respectively,
\begin{equation}\tag{\ref{eq:rafc3}$'$}\label{eq:rafc3'}
\rho_F(r_j)\;\setminus\; F \quad \subseteq \quad \bigcup_{r \;|\;r\;\prec\;r_j} \pi_{net}(r).
\end{equation}
\end{prop}
\begin{proof}
Assume \eqref{eq:rafc2} (resp. \eqref{eq:rafc3}) and let $x_m$ be an
\textsc{F-product} but not an \textsc{F-reactant} of a reaction $r_{j'}$,
$\sigma^+_{mj'}>0$ and $\sigma^-_{mj'}=0$. Via Eq.~\eqref{eq:constraintsF}
we obtain that
$s^+_{mj'}\;=\;\sigma^+_{mj}+k_{mj'}\; >\;
k_{mj'}\;=\;\sigma^-_{mj'}+k_{mj'}\;=\;s^-_{mj'}$,
and thus any such $x_m$ is a net-product of $r_{j'}$ and \eqref{eq:rafc2'}
(resp. \eqref{eq:rafc3'}) follows.

Conversely, let $x_m$ be an \textsc{F-reactant} of a reaction $r_j$, which
is not in the food-set $F$, and assume \eqref{eq:rafc3'}: then $x_m$ is a
net-product of a reaction $r_{j'}$ with $r_{j'}\prec r_j$. Either $x_m$ is
not an \textsc{F-reactant} of $r_{j'}$ or it is. If the latter holds, then
we can iteratively apply a second time \eqref{eq:rafc3'} for $r_{j'}$, and
conclude that there exists another reaction $r_{j''} \prec r_{j'}$ so that
$x_m$ is a net product of $r_{j''}$. By iteration of the same procedure and
finiteness of $R$,  $x_m \not\in F$ implies the existence of a reaction
$r_{\bar{j}}$ for which $x_m$ is a net-product but not an
\textsc{F-reactant}. Via Prop.~\ref{prop:inclusions}, $x_m$ is an
\textsc{F-product} of $r_{\bar{j}}$ and thus \eqref{eq:rafc3} follows.

Finally, let $x_m$ be an \textsc{F-catalyst} of a reaction $r_j$, which is
not in the food set $F$, and assume \eqref{eq:rafc2'}: $x_m$ is a
net-product of a reaction $r_{j'}\in R$, and it may be an
\textsc{F-reactant} of $r_{j'}$ or not. In the latter case, again via
Prop.~\ref{prop:inclusions}, \eqref{eq:rafc2} follows. In the former case,
we can then apply \eqref{eq:rafc3'}, or equivalently \eqref{eq:rafc3},
to obtain that $x_m$ is an \textsc{F-product}, but not an
\textsc{F-reactant} of a reaction $r_{j''}$ with $r_{j''}\prec r_{j'}$ and
\eqref{eq:rafc2} follows again.
\end{proof}

Note that the same iterative argument can be used to state the conditions
in terms of \textsc{F-products} $\pi_F(r)$, rather than more restrictive
net-products $\pi_{\text{net}}$. However, net-products have the advantage
that they are always identified by the stoichiometric matrix $S$ alone, and
this will be used in the remainder of this paper.

\begin{remark}[Original RAF framework II] The 
definition of a \textsc{RAF} in the original framework \emph{\cite{hordijk04}} is different, as it involves the so-called \emph{closure set}, via $R$, of the food set. For simplicity of presentation, we have stated the equivalent Def.~\ref{def:RAF}. The equivalence of such definitions has already been proved in the literature, e.g. in \emph{\cite{steel13minimal}}.
\end{remark}

We proceed by defining a special type of \textsc{RAF} sets, `\emph{constructively autocatalytic}', which were originally introduced in \cite{mossel05CAF}.

\begin{definition}\label{def:CAF}
A \textsc{Constructively Autocatalytic Food-Generated set}, \textsc{CAF} in short, is a reaction network $\Gamma_F=(X,R,F)$ with a specified \emph{food set} $F\subseteq X$ such that the following properties hold:
\begin{enumerate}[label=C\arabic*.]
\item For each reaction $r\in R$, the set of \textsc{F-catalysts} is nonempty:
$\kappa_F(r)\neq\emptyset,$ for all $r\in R$.  
\item There exists a partial order $\sqsubseteq$ on $R$ such that both \textsc{F-reactants} and \textsc{F-catalysts} of a reaction $r_j\in R$ either are in the food set or they are \textsc{F-products}, but not \textsc{F-reactants}, of at least one reaction $r_{j'}$ with $r_{j'}\sqsubset r_{j}$. In symbols,
\begin{equation}\label{eq:cafc2}
\rho_F(r_j) \; \cup\; \kappa_{F}(r_j)\;\setminus\; F \quad \subseteq \quad \bigcup_{r \;|\;r\;\sqsubset\;r_j} \bigg(\pi_{\text{F}}(r) \;\setminus \; \rho_{\text{F}}(r)\bigg).
\end{equation}
\end{enumerate}
\end{definition}

It is clear that any \textsc{CAF} is also a \textsc{RAF}, but the reverse
is not true in general. More specifically, Def.~\ref{def:CAF} of a
\textsc{CAF} strengthens condition R2 of Def.~\ref{def:RAF} of a RAF in
requiring that also the \textsc{F-catalysts} must be produced by reactions
that occur earlier in the partial order $\sqsubseteq$, which we refer to as
\textsc{CAF order}. By definition, note that the \textsc{CAF order}
$\sqsubseteq$ is an \emph{extension} \cite{jech:1973} of the
\textsc{generation order} $\preceq$ in the sense that
\[
    r_{j'}\prec r_{j} \quad \Rightarrow \quad r_{j'} \sqsubset r_{j},
\]
while the example in~\ref{sect:examples}.\ref{egg:3} shows that the converse implication is not true. 

We recall that the union of the \textsc{F-reactants} and the \textsc{F-catalysts}
of a reaction $r$ just constitute its reactants $\rho(r)$, and therefore
the \textsc{CAF-order} $\sqsubseteq$ gives a general constructibility
condition of the entire reaction (with no distinction among its reactants)
from the food set. Clearly, via the same argument as
Prop.~\ref{prop:equivnet}, also the right-hand side of property
\eqref{eq:cafc2} can be equivalently restated in terms of net-reactants,
yielding:
\begin{equation}
  \rho (r_j)\;\setminus\; F \quad \subseteq \quad \bigcup_{r
    \;|\;r\;\sqsubset \;r_j} \pi_{net}(r).
  \tag{\ref{eq:cafc2}$'$}\label{eq:cafc2'}
\end{equation} 
In particular, \eqref{eq:cafc2'} shows that the \textsc{CAF} construction
does not even require defining \textsc{F-reactants} and
\textsc{F-products}: it only requires - via property C1 in
Def.~\ref{def:CAF} - that each reaction has at least one
\textsc{F-catalyst}.

A note of caution is in order here. In the literature of RAFs, a sub-RAF
$\Gamma'=\Gamma_F$ is typically a RAF $\Gamma'=(X(R'),R',F\cap X(R'))$ with
$R' \subset R$, i.e., where the species set is chosen to be the support of
reactions in $R'$, and same food set $F$. To avoid confusion with the more
flexible definition of subnetworks given in Sec.~\ref{sec:preliminaries},
we will always fully specify the species and reaction sets of the subnetwork.

To continue in this direction, if we fix a food-set $F\subset X$, the
following proposition characterizes the situation when a \textsc{RAF}
$\Gamma_F$ contains a subnetwork $\Gamma'_{F}$, which is a \textsc{CAF} for
the same food set.

\begin{prop}\label{prop:rafwocaf}
  Let $F\subseteq X$ be a fixed food-set.  
  A \textsc{RAF} $\Gamma_F=(X,R,F)$
  contains a subnetwork $\Gamma'_F=(X(R'),R',F\cap X(R'))$ that is a
  \textsc{CAF} if and only if there exists at least one reaction $r_j\in
  R$ for which all reactants are in the food set, i.e., $x_m\in F$ for all
  $x_m$ with $s^-_{jm}>0$.
\end{prop}
\begin{proof} $(\Rightarrow)$ Assume that the \textsc{RAF} $\Gamma_F=(X,R,F)$
contains a \textsc{CAF} $\Gamma'_F=(X(R'),R',F\cap X(R'))$, and let $r_j$
be one of the minima in $R'$ according to the \textsc{CAF order}
$\sqsubseteq$, i.e.  $ r_j\in \operatorname{min}_\sqsubseteq(R')$. By
\eqref{eq:cafc2} in Def.~\ref{def:CAF}, all reactants of $r_j$ are in the
food set $F$, since $\{r \; | \; r \sqsubset r_j\}=\emptyset$ by
construction.

$(\Leftarrow)$ Assume that the \textsc{RAF} $\Gamma_F=(X,R,F)$ contains a
reaction $r_j$ for which all reactants are in the food set, then the
subnetwork: $\Gamma_F'=(X(r_j),r_j,F\cap X(r_j)$, made of the single
reaction $r_j$ is a \textsc{CAF}, since it possesses an \textsc{F-catalyst}
and all of its reactants are in the food set.
\end{proof}

\section{Main Result}

Although the two frameworks of \emph{stoichiometric autocatalysis} and
\textsc{Reflexively Autocatalytic Food-generated sets} (RAFs) appear
superficially distinct, a \textsc{RAF} -- as any reaction network --
nevertheless possesses an associated stoichiometric matrix. Thus, the
question whether such stoichiometric matrix of a \textsc{RAF} contains any
stoichiometrically autocatalytic submatrix is well-posed. To answer this
question positively, we state the main result of this paper, which
guarantees that any \textsc{RAF}, which is not itself a \textsc{CAF}, is indeed always
stoichiometrically autocatalytic in the sense of Def.~\ref{def:autocat}.

\begin{theorem}\label{thm:main}
  The stoichiometric matrix $S$ of any \textsc{RAF} $\Gamma_F=(X,R,F)$ that
  is not itself a \textsc{CAF} possesses a well-formed semipositive
  submatrix $S^*=S[X^*,R^*]$ with the following properties: 
\begin{enumerate}[label=\emph{T\arabic*.}]
\item No species in $X^*$ are in the food set, i.e. $X^*\cap F=\emptyset$.

\item $S^*$ is square. Specifically,
there exist orderings $X^* = \{x_{(1)},\dots,x_{(|X^*|)}\}$ and $R^* = \{r_{(1)},\dots,r_{(|R^*|)}\}$ s.t. alternating elements form an elementary circuit $\mathcal{C}$,
\begin{equation*}
\mathcal{C}:\hspace{6pt} 
r_{(1)} - x_{(1)} - r_{(2)} - x_{(2)} - \cdots - r_{(|R^*|)} - x_{(|X^*|)} - r_{(1)},
\end{equation*}
where, for each $i$ (indices modulo $|X^*|=|R^*|$), $x_{(i)}$ is a reactant of $r_{(i)}$ and a net-product of $r_{(i+1)}$. 
\item For any \textsc{F-reactant} $x_m$ of a reaction $r_j\in R^*$, there
  exists a reaction $r_{j'}\in R^*$ with $r_{j'} \prec r_j$ and
  $S^*_{mj'}>0$. In particular, any reaction $r^*\in R^*$ that is minimal
  in $R^*$ with respect to the \textsc{generation order} $\preceq$ has no
  \textsc{F-reactants} in $X^*$ and thus corresponds to a nonnegative
  column $S^*[X^*,r^*] \ge 0$.
\end{enumerate}
\end{theorem}
\begin{proof}
Consider the maximal subnetwork 
\[\Gamma'_F=(X(R'),R',F\cap X(R'))\subseteq
\Gamma_F,\] which is a \textsc{CAF}. The fact that such a maximal network is
well-defined and unique is known in the literature, see for example
\cite{steel20} where the notation \textsc{max}CAF is used. To keep this
contribution self-contained, we nevertheless construct such a set from
scratch as follows. Let $R'_{(1)}$ be the set of reactions $r_{(1)}$ whose
reactants are all in the food set $F$. Iteratively, let $R'_{(n)}$ be the
set of reactions $r_{(n)}$ whose reactants are either in the food set, or
are net-products of at least one reaction $r_{(m)}$ in $R'_{(m)}$, with
$m<n$. Due to finiteness of $R$, this process is finite, and we can
consider $R'\coloneqq \bigcup_i R'_{(i)}$. By construction, $\Gamma'_F$ is a
maximal \textsc{CAF} in $\Gamma_F$, i.e., the reactions of any other
\textsc{CAF} in $\Gamma_F$ with food set in $F$ are contained in $R'$ . By
assumption, moreover, $\Gamma'_F \subsetneq \Gamma_F$, and in particular
$R''\coloneqq (R\setminus R')\neq \emptyset$.

Let $X''\coloneqq X(R'')$ and $F''\coloneqq F\cup X(R')$. The network
$\Gamma''_{F''} = (X'', R'',F''\cap X'') $ is a \textsc{RAF} itself. In
fact, property R1 and property R2 in Def.~\ref{def:RAF} directly follow
from the same property R1 and property R2 of the \textsc{RAF}
$\Gamma_F$. The \textsc{generation order} of $\Gamma''_{F''}$ also follows
from the \textsc{generation order} of $\Gamma_{F}$, with a translation of
the values of the \textsc{generation map} $\gamma(\Gamma'')$ in
$\Gamma''_{F''}$, with respect to $\gamma(\Gamma)$, due to enlargement of
the food set $F''=F\cup X(R')$ and removal of reactions associated to the
\textsc{CAF} $\Gamma'_F$, since $R''=R\setminus R'$. In particular,
\textsc{first-generation reactions} in $\Gamma''_{F''}$ are the ones with
\textsc{F-reactants} in the food set $F''$, i.e., either in the original
food set $F$ or in $X'$, and so on. Moreover, by construction,
$\Gamma''_{F''} = (X'', R'',F''\cap X'') $ contains no subnetwork that is a
\textsc{CAF} and in particular, via Prop.~\ref{prop:rafwocaf}, any reaction
$r_j$ in $R''$ has at least one reactant $x_m$, which is not in the food
set $F''$.

We now introduce two maps $\chi$ and $\psi$. The map $\chi$
assigns to each reaction $r_j\in R''$ a non-food species
$\chi(r_j)=x_m\in X''\setminus F''$, which is a reactant of $r_j$, in symbols
\begin{equation}\label{eq:chi}
\chi(r_j)=x_m\in \rho(r_j)\setminus F''.
\end{equation}
The map $\chi$ is well-defined by the construction of the \textsc{RAF}
$\Gamma''_{F''}$.

The map $\psi$ is defined on the set
$\bigcup_{r\in R''}\rho(r)\setminus F'' \subset X''$.  It assigns to each
species $x_m$ in such set a reaction $r_j\in R''$ that is minimal with
respect to the \textsc{generation order} among those reactions producing
$x_m$ as a net-product, in symbols:
\begin{equation}\label{eq:psidef}
  \psi(x_m)\in
  \min_{\preceq}\{\,r_j\in R'' \mid x_m \in \pi_{\mathrm{net}}(r_j)\,\}.
\end{equation}
The map $\psi$ is well-defined by the definition of a \textsc{RAF}, and it
also follows that $\psi(x_m)$ cannot have $x_m$ as \textsc{F-reactant},
since otherwise we would have that $x_m$ is a net product of a reaction
$r\in R''$ with $r\prec \psi(x_m)$, by property R3 in Def.~\ref{def:RAF}
and Prop.~\ref{prop:equivnet}, contradicting \eqref{eq:psidef}.

Let $r_j\in R''$. By alternate application of $\chi$ and $\psi$, we build a
sequence $(\Lambda_k)_{k\ge 0}$ of alternating reactions in $R''$ and
species in $X''\setminus F''$:
\[
  \Lambda_0 = r_j, \qquad
  \Lambda_{k+1} =
  \begin{cases}
    \chi(\Lambda_k), & \text{if $\Lambda_k\in R''$,}\\
    \psi(\Lambda_k), & \text{if $\Lambda_k\in X'' \setminus F''$.}
  \end{cases}
\]
Let $\mathcal{C}$ be the elementary circuit identified by the sequence
$(\Lambda_k)$, that is, the subsequence
$\mathcal{C} = (\Lambda_h,\Lambda_{h+1},\ldots,\Lambda_{i-1})$,
where $h<i$ are the smallest indices in $(\Lambda_k)$ such that
$\Lambda_i=\Lambda_h$. Note that due to the finiteness of $R''$ and $X''$,
we always find one such circuit, starting from any reaction
$r_j$. Different starting reactions, as well as different choices for
$\chi$ and $\psi$, may however lead to different circuits.

Let $X^*$ and $R^*$ be the species and reactions touched by
$\mathcal{C}$. By construction of the map $\chi$,
$X^* \cap F'' = \emptyset$, and in particular $X^*\cap F=\emptyset$, and
property T1 is satisfied. Moreover, $|R^*|=|X^*|$ and $\mathcal{C}$
provides the orderings for property T2. The maps $\psi$ and $\chi$ define
in particular bijections between $X^*$ and $R^*$, and guarantee that $x$ is
a net-product of $\psi(x)$ and a reactant of $\chi^{-1}(x)$, and property T2 is satisfied. 
$r^*\in R^*$ has a reactant $\chi(r^*)\in X^*$ and a product
$\psi^{-1}(r^*)\in X^*$. Thus the subnetwork $\Gamma^*(X^*,R^*)$ is
well-formed; and so is the associated stoichiometric matrix
$S^*=S[X^*,R^*]$.

To show T3, we define a partition $R^*=\bigcup_{i\ge 1} R^*_{(i)}$
on $R^*$ as follows:
\[
  R^{*}_{(i)}\coloneqq \begin{cases}
    \operatorname{min}_{\preceq}(R^*) \quad\quad &\text{for $i=1$};\\
    \\
    \operatorname{min}_{\preceq}
    \bigg(R^* \setminus \bigcup_{n=1}^{i-1}R^*_{(n)}\bigg)\quad
                                                     &\text{for $i\ge2$}.
  \end{cases}
\]
We show that for any reaction $r^*_{(n)}\in R^{*}_{(n)}$ all
\textsc{F-reactants} in $X^*$ are net products of reactions
$r^*_m\in R^*: r^*_m\prec r^*_n$., i.e.
$(\rho_F(r^*_N)\cap X^*)\subseteq \bigcup_{i=1}^{n-1} \pi_{\text{net}}
(R^*_{(i)})$. In fact, indirectly assume that there exists
$x_m\in X^*\setminus \bigcup_{i=1}^{n-1} \pi_{\text{net}} (R^*_{(i)})$
which is an \textsc{F-reactant} of $r^*_{(n)}$, i.e.
$x_m\in \rho_{\text{F}}(r^*_{(n)})$. This, in particular, implies that
$r^*_{(n)} \preceq \psi(x_m)$, since $x_m$ is a net-product of
$\psi(x_m)$. By Def.~\ref{def:RAF} of \textsc{RAF}, however, there exists a
reaction $\bar{r}\in R''$ with $\bar{r}\prec r^*_{(n)}$, where $x_m$ is a
net product. However, we have reached a contradiction, since the reaction
$\psi(x_m) \in R^*$ is defined via \eqref{eq:psidef} as belonging to the
minima over $R''$ with this property, in symbols:
\[
  \psi(x_m)\preceq \bar{r}\prec r^*_{(n)} \preceq \psi(x_m).
\]
Thus $r^*_{(n)}\in R^{*}_{(n)}$ has no \textsc{F-reactants} in
$X^*\setminus \bigcup_{i=1}^{n-1} \pi_{\text{net}} (R^*_{(i)})$. In
particular, by construction, $R^{*}_{(1)}$ has no \textsc{F-reactants} in
$X^*$, and thus no net-reactants by Prop.~\ref{prop:inclusions}. The
associated column $S[X^*,r^*]$ is therefore nonnegative, and property T3 is
satisfied.

Finally, we prove semipositivity of $S^*$. We note that no row in $S^*$ is
identically zero, since any $x_m\in X^*$ is net-product of
$\psi(x_m)\in R^*$. We can then leverage the above construction (property
T3) and define $v(N)\in \mathbb{R}^{|R^*|}_{>0}$ as:
\[v_j(N)\coloneqq\dfrac{N}{i} \quad\quad\quad \text{for $r_j\in R^*_{(i)}$},
\]
where $N>0$ is a positive number. By construction, $v_{j'}(N)-v_j(N) \to
+\infty$, for $r_{j'}\prec r_{j}$ and $N\to+\infty$. This implies that
$S^*v(N) >0$ for $N$ sufficiently large, and thus $S^*$ is semipositive.
\end{proof}

Before we proceed, a few explanatory remarks regarding 
  Thm.~\ref{thm:main} are in order.

\begin{remark}\label{rmk:thm1}
  If all species $X$ are in the food-set, i.e. $X=F$, then T1 in
  Thm.~\ref{thm:main} cannot hold. This is consistent with the statement of
  the theorem because any \textsc{RAF} where all species are in the food
  set is trivially a \textsc{CAF}, and thus Thm.~\ref{thm:main} does not
  apply.
\end{remark}

\begin{remark}\label{rmk:main4}
  The inverse sequence of the circuit $\mathcal{C}$, i.e., \[ r_1 \to x_{|X^*|} \to r_{|R^*|} \to \ldots \to x_1 \to r_1
  \]
  identifies a directed elementary circuit in the bipartite graph
  representation of the associated reaction network. In fact, by
  construction, the direction $\to$ identifies for a pair $(x_i \to r_i)$ a
  reactant/reaction relation and for a pair $(r_{i+1}\to x_{i})$ a
  reaction/product relation. 
\end{remark}

\begin{remark}\label{rmk:main5}
  Condition \emph{T2} and condition \emph{T3} jointly imply that any
  reaction $r^*\in R^*$ which is minimal in $R^*$ with respect to the
  \textsc{generation order} $\preceq$ has at least one \textsc{F-catalysts}
  but no \textsc{F-reactants} in $X^*$.
\end{remark}

\begin{remark}
  A \textsc{RAF} that is not itself a \textsc{CAF} may strictly contain a
  \textsc{CAF} as shown in~\ref{sect:examples}.\ref{egg:3}. Prop.~\ref{prop:rafwocaf} gives a stronger
  condition that characterizes the case of a \textsc{RAF} not containing
  any \textsc{CAF}. This condition is simply that any reaction $r$ has at
  least one reactant not in the food set. Any \textsc{RAF} where this
  condition applies is thus stoichiometrically autocatalytic, see also
Cor.~\ref{cor:main}.
\end{remark}

The statement of Thm.~\ref{thm:main} is rather technical. Its most direct consequence, 
  however, connects RAF theory to
  stoichiometric autocatalysis:
\begin{cor}\label{cor:main}
  Any RAF that is not itself CAF is stoichiometrically autocatalytic.
\end{cor} 
\begin{proof}
  Since the matrix $S^*=S[X^*,R^*]$ is semipositive and
  well-formed, it is autocatalytic. By definition, thus, the
  network $\Gamma=(X,R)$ is stoichiometrically autocatalytic.
\end{proof}

Further, Prop.~\ref{prop:equivnet} and condition T3 in Thm.~\ref{thm:main} can
be used to identify a \textsc{RAF} based on species in $X^*$ and $R^*$
as the following corollary states.
\begin{cor}\label{cor:subraf}
  Any \textsc{RAF} $\Gamma_F=(X,R,F)$ that is not itself a \textsc{CAF}
  possesses a subnetwork $\Gamma^*_{F^*}=(X(R^*),R^*,F^*)$, with food set
  $F^*=X(R^*) \setminus X^*$, which is itself a \textsc{RAF} and whose
  stoichiometric matrix $S^*[X^*,R^*]$, restricted to the non-food
  molecules, is square, well-formed, and semipositive.
\end{cor}
\begin{proof}
  Consider the subnetwork $(X(R^*),R^*,F^*)$, where $R^*$ is constructed in the proof of
    Thm.~\ref{thm:main}. We show that properties R1, R2, and R3 in
    Def.~\ref{def:RAF} are satisfied:\\
  (R1) Since $R^* \subseteq R$, and $\Gamma_F$ is a \textsc{RAF},
  then any reaction $r\in R^*$ has at least one \textsc{F-catalyst}. As the
  species set of $\Gamma_{F^*}$ is the support of $R^*$, property R1
  follows.\\
  (R2) By T2 in Thm.~\ref{thm:main}, the \textsc{F-catalysts} of
  $R^*$ either are in the food set or are net-products of reactions in
  $X^*$, which via construction in the proof of Thm.~\ref{thm:main} satisfies $X^*\subseteq X(R^*)$, and R2 follows.\\
  (R3) By T3 in Thm.~\ref{thm:main}, the \textsc{F-reactants} of a reaction
  $r_j\in R^*$ are either in the food set $F^*$ or net-products of at least
  one reaction $r_{j'} \in R^*$ with $r_{j'}\prec r_j$.
\end{proof}

The statement can be further strengthened by assuming that each reaction
$r$ possesses at least on \textsc{F-catalyst} not in the food set, as the
following corollary states.
\begin{cor}\label{cor:subrafcat}
  Let $\Gamma_F=(X,R,F)$ be a \textsc{RAF} and assume that each reaction
  $r\in R$ possesses at least one \textsc{F-catalyst} not in the food
  set. Then $\Gamma_F=(X,R,F)$ possesses a subnetwork
  $\Gamma^*_{F^*}=(X(R^*),R^*,F^*)$, with food set
  $F^*=X(R^*) \setminus X^*$, which is itself a \textsc{RAF} and whose
  stoichiometric matrix $S^*[X^*,R^*]$, restricted to the non-food
  molecules, is square, well-formed, and semipositive. Moreover, there is a
  bijection $\chi:R^*\mapsto X^*$ such that $\chi(r)$ is an
  \textsc{F-catalyst} \mbox{of $r$.}
\end{cor}
\begin{proof}
  By Prop.~\ref{prop:rafwocaf}, the \textsc{RAF} $\Gamma_F$ does not
  contain any \textsc{CAF}. Operating as in the proof of
  Thm.~\ref{thm:main}, then, $\Gamma''_{F''}=\Gamma_F$ and we consider a
  stronger version $\tilde{\chi}$ of the map $\chi$ defined in
  \eqref{eq:chi}.  The map $\tilde{\chi}$ assigns to each reaction $r_j\in
  R$ a non-food species $\chi(r_j)=x_m\in X\setminus F$, which is an
  \textsc{F-catalysts} of $r_j$, i.e. $\chi(r_j)=x_m\in
  \kappa_F(r_j)\setminus F$.  Then the proof proceeds identical to the one
  of Cor.~\ref{cor:subraf}, with the only addition that in this case
  $\tilde{\chi}$ represents the bijection as stated.
\end{proof}

\section{Examples}
\label{sect:examples} 

For simplicity of visualization, in what follows we will use the bold
notation $\mathbf{x}_i$ for molecules in the food set $F$.

\begin{eggsample}[Inferring the \textsc{RAF} structure may not require \textsc{F-sets}
  specification.]
\label{egg:ex1} 
First, we ask whether the \textsc{RAF} structure can be inferred from the
reaction network structure alone.  We positively answer this question here in~\ref{sect:examples}.\ref{egg:ex1} for two simple networks and clarify in~\ref{sect:examples}.\ref{egg:ex2} that the answer is negative in general.
Consider a single reaction of the following form:
\[
  \mathbf{x}_1+ x_2 \underset{r_1}{\longrightarrow} 2x_2. 
\]
Clearly only one RAF $\Gamma_F=(\{x_1,x_2\},r_1,\{x_1\})$ is possibly
associated with such a reaction, namely:
\[
  \mathbf{x}_1+ (x_2) \underset{r_1}\longrightarrow x_2 + (x_2). 
\]
In particular, $\mathbf{x}_1$ is the single \textsc{F-reactant}, $x_2$ is
both \textsc{F-product} and \textsc{F-catalyst} of $r_1$. The circuit
$\mathcal{C}$ from Thm.~\ref{thm:main} in this case is simply
\[
   \mathcal{C}= r_1 - x_2 - r_1,
\]
as $x_2$ is both an \textsc{F-catalyst} and an \textsc{F-product} of
$r_1$. Stoichiometric autocatalysis is thus seen in the species $x_2$, with
associated stoichiometric matrix:
\[
    S^*=S[x_2,r_1]=2-1=1>0,
\]
which is indeed semipositive and well-formed. 

Another example that showcases a similar situation is the following.
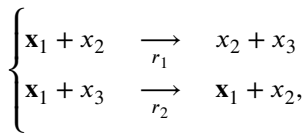
\[
\begin{cases}
    \mathbf{x}_1 + x_2 \quad \underset{r_1}\longrightarrow  \quad x_2 + x_3\\
    \mathbf{x}_1 + x_3 \quad \underset{r_2}\longrightarrow \quad \mathbf{x}_1 + x_2,
\end{cases}
\]
with associated unique RAF $\Gamma_F=(\{x_1,x_2,x_3\},\{r_1,r_2\},\{x_1\})$, namely:
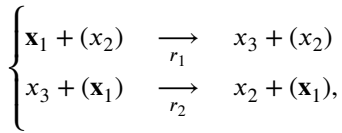
\[
\begin{cases}
    \mathbf{x}_1 + (x_2) \quad\underset{r_1}\longrightarrow \quad  x_3 + (x_2)\\
    x_3 + (\mathbf{x}_1) \quad \underset{r_2}\longrightarrow  \quad x_2 + (\mathbf{x}_1),
\end{cases}
\]
with \textsc{F-sets} specification as:
\begin{center}
\begin{tabular}{c|c c c}
 & \textsc{F-reactants} & \textsc{F-products} & \textsc{F-catalysts}\\
\hline
 $r_1$        &  $\mathbf{x}_1$ & $x_3$ & $x_2$ \\
  $r_2$       & $x_3$ & $x_2$ & $\mathbf{x}_1$
\end{tabular}.
\end{center}
Here, again, there is a unique choice of the circuit $\mathcal{C}$ of Thm.~\ref{thm:main}, i.e.
\[
  \mathcal{C}=  r_1 - x_2 - r_2 - x_3 - r_1,
\]
with associated autocatalytic stoichiometric matrix:
\[
    S^*=\begin{pmatrix}
        0 & 1\\
        1 & -1
    \end{pmatrix}.
\]
\end{eggsample}

\begin{eggsample}[Inferring the \textsc{RAF} structure may require \textsc{F-sets}
  specification.]
\label{egg:ex2} In contrast to the simple examples in~\ref{sect:examples}.\ref{egg:ex1}, even one single
reaction may be ambiguous with respect to its RAF structure, without
explicit specification of its \textsc{F-reactants}, \textsc{F-products} and
\textsc{F-catalysts}. Consider the following reaction
\[
2\mathbf{x}_1+x_2 \quad \underset{r_1}\longrightarrow \quad \mathbf{x}_1+2x_2 .
\]
One basic question is:
\begin{center}
    Which \textsc{RAF} is defined on $\Gamma_F=(\{\mathbf{x}_1,x_2\},r_1,\{\mathbf{x}_1\})$?
\end{center}
The question is not well-posed. By Prop.~\ref{prop:inclusions}, we know
that $\mathbf{x}_1$ is an \textsc{F-reactant} since it is a net-reactant
and $x_2$ is an \textsc{F-product} since it is a net-product. A priori, we
can further say that, for $\Gamma_F$ to be a RAF, $x_2$ cannot be an
\textsc{F-reactant} and must consequently be an \textsc{F-catalyst}. Still,
we are left with two valid choices for $x_1$, i.e., whether $x_1$ is also
an \textsc{F-catalyst} or not. The former choice leads to a \textsc{RAF} as
follows:
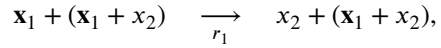
\[
\mathbf{x}_1+ (\mathbf{x}_1 + x_2) \quad \underset{r_1}\longrightarrow \quad x_2 + (\mathbf{x}_1+x_2),
\]
while the latter leads to:
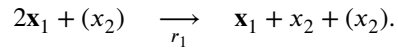
\[
2\mathbf{x}_1+ (x_2) \quad \underset{r_1}\longrightarrow \quad \mathbf{x}_1+ x_2 + (x_2).
\]
The fact that such reaction is stoichiometrically autocatalytic is
independent of the above choice: in both cases, we get that circuit
$\mathcal{C}$ is
\[
   r_1- x_2 - r_1, 
\]
with associated 1-dimensional autocatalytic stoichiometric matrix:
\[
    S^*=S[x_2,r_1]=1.
\]
\end{eggsample}



\begin{eggsample}[\textsc{RAF} vs \textsc{CAF} order and stoichiometric
    autocatalyticity of \textsc{CAF}s.]
  \label{egg:4}
Consider the following \textsc{CAF}
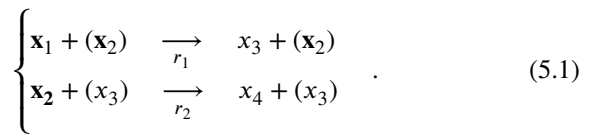
\begin{equation}\label{ex:2:CAF}
\begin{cases}
        \mathbf{x}_1 + (\mathbf{x}_2) \quad\underset{r_1}\longrightarrow\quad x_3 + (\mathbf{x}_2)\\
        \mathbf{x_2} + (x_3) \quad\underset{r_2}\longrightarrow \quad x_4 + (x_3)
    \end{cases}.
\end{equation}
Since any \textsc{CAF} is in particular a \textsc{RAF}, we have both the \textsc{CAF order} and \textsc{RAF generation order} on $\{r_1,r_2\}$. However, they are not the same. In fact, according to the \textsc{CAF} order, 
\[r_1 \sqsubset r_2,
\]
since the reactant $x_3$ of $r_2$ is not in the food set, but it is produced by reaction $r_1$. On the other hand, $x_3$ is an \textsc{F-catalyst} of $r_2$, and hence for the \textsc{generation order} $\preceq$ both reactions are of \textsc{first generation}:
\[
    \gamma(r_1) = \gamma(r_2) =1,
\]
since all the \textsc{F-reactants} of both reactions are in the food set. Example \eqref{ex:2:CAF}, being a \textsc{CAF}, also shows that Thm.~\ref{thm:main} does not necessarily hold for a \textsc{RAF}, which is a \textsc{CAF}. Indeed, the stoichiometric matrix 
\[
    S=\begin{pmatrix}
        -1 &  0\\
        0  & -1\\
        1  &  0\\
        0  & 1
    \end{pmatrix}
\]
does not possess any autocatalytic submatrix and such \textsc{CAF} is thus not stoichiometrically autocatalytic.

In turn, a CAF can be as well stoichiometrically autocatalytic as, for example,
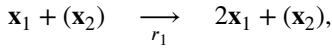
\[
\mathbf{x}_1+(\mathbf{x}_2)\quad \underset{r_1}\longrightarrow\quad  2\mathbf{x}_1 + (\mathbf{x}_2),
\]
which is stoichiometrically autocatalytic in the species $\mathbf{x}_1$.
\end{eggsample}

\begin{eggsample}[A \textsc{RAF} that strictly contains a \textsc{CAF}.]
  \label{egg:3}
  The following \textsc{RAF}
  \[ \begin{cases}
    \mathbf{x_1} + (\mathbf{x_2}) \quad \overset{r_1}\longrightarrow \quad x_3 + (\mathbf{x_2}) \\
    x_3 + (x_5) \quad \overset{r_2}\longrightarrow \quad x_4 + (x_5) \\
    x_4 + (x_3) \quad \overset{r_3}\longrightarrow \quad x_5 + (x_3) \\
  \end{cases}\]
  contains, as a subnetwork, the \textsc{CAF}
  \[
  \mathbf{x_1} + (\mathbf{x_2}) \quad \overset{r_1}\longrightarrow \quad x_3 + (\mathbf{x_2}),
  \]
  but it is not a \textsc{CAF} itself, since reaction $r_2$ is catalyzed by
  the species $x_5$, which is only produced by a following (in
  \textsc{generation order}) reaction $r_3$. Consistently, it is
  stoichiometric autocatalytic via the following matrix $S^*$:
  \[S^*=S[\{x_4,x_5\},\{r_2,r_3\}]=\begin{pmatrix}
  1 & -1\\
  0 &  1
  \end{pmatrix},\]
  which is indeed autocatalytic, being semipositive and well-formed. The
  associated circuit $\mathcal{C}$ per Thm.~\ref{thm:main} is
  \[
  \mathcal{C}=r_2 - x_5 - r_3 - x_4 -r_2.
  \]
\end{eggsample}

\begin{eggsample}[Many \textsc{subRAFs} and only one $S^*$.]
  \label{egg:5}
  It is not possible to infer the number of distinct \textsc{subRAFs} from
  the number of matrices $S^*$ satisfying T1-T2-T3 in
  Thm.~\ref{thm:main}. In fact, consider the following \textsc{RAF}
  $\Gamma_F=(X,R,F)$:
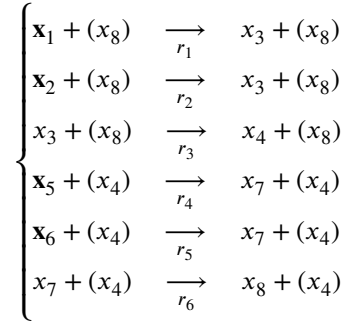
\[
    \begin{cases}
        \mathbf{x}_1 + (x_8) \quad\underset{r_1}\longrightarrow \quad x_3 +
        (x_8)\\ \mathbf{x}_2 + (x_8)\quad\underset{r_2}\longrightarrow
        \quad x_3 + (x_8)\\ x_3 + (x_8) \quad \underset{r_3}\longrightarrow
        \quad x_4 + (x_8)\\ \mathbf{x}_5 +
        (x_4)\quad\underset{r_4}\longrightarrow \quad x_7 +
        (x_4)\\ \mathbf{x}_6 + (x_4)\quad\underset{r_5}\longrightarrow
        \quad x_7 + (x_4)\\ x_7 + (x_4) \quad \underset{r_6}
        \longrightarrow \quad x_8 + (x_4)
    \end{cases}
\]
It can readily be seen that any subnetwork $\Gamma'_F=\{X(R'),R',$ $F\cap X(R')\}$ made of reactions $R'$ that include $r_3,$ $r_6$, and  either $r_1$ or $r_2$ (or both) and either $r_4$ or $r_5$ (or both) is itself a \textsc{RAF}. The number of such \textsc{subRAFs} is $3^2=9$. Nevertheless, all of them share the same, unique, stoichiometric submatrix $S^*$, which is autocatalytic and satisfies T1-T2-T3. This matrix is defined over species $X^*=\{x_4,x_8\}$ and reactions $R^*=\{r_3,r_6\}$:
\[S[X^*,R^*]=\begin{pmatrix}
        1 & 0\\
        0 & 1
    \end{pmatrix},
\]
with associated circuit $\mathcal{C}$:
\[ r_3- x_8 - r_6 - x_4 - r_3.
\]
In this example, we can as well apply Cor.~\ref{cor:subraf} to define an enlarged food set $F^*=X(R^*)\setminus X^*=\{x_3,x_7\}$, and the associated \textsc{RAF} $\Gamma^*_{F^*}=(X(R^*),R^*,F^*)$, which reads:
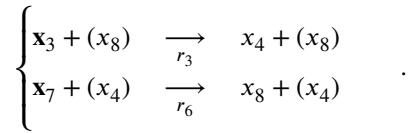
\[
    \begin{cases}
        \mathbf{x}_3 + (x_8) \quad \underset{r_3}\longrightarrow \quad x_4 + (x_8)\\       
        \mathbf{x}_7 + (x_4) \quad \underset{r_6} \longrightarrow \quad x_8 + (x_4)
    \end{cases}\quad.
\]
\end{eggsample}
\begin{eggsample}[Only one \textsc{RAF} and many $S^*$.]
  \label{egg:6}   Similar in spirit to \ref{sect:examples}.\ref{egg:4}, but complementary, the construction of
  $S^*$ in the proof of Thm.~\ref{thm:main} is not unique. It is based on
  two maps $\chi : R \to X$ and $\psi : X \to R$. The map $\chi$ assigns to
  each reaction $r$ one of its \textsc{F}-catalysts (if not in the food
  set), while $\psi$ assigns to each species $x$ a reaction $r$ in which
  $x$ is a net product, chosen minimally with respect to the
  \textsc{generation order}. This construction involves two sources of
  arbitrariness: the choice of an \textsc{F}-catalyst in $\chi$ when a
  reaction admits more than one, and the choice of a reaction in $\psi(x)$
  when $x$ is produced by several reactions of the same minimal generation.
  It is therefore possible to have one \textsc{RAF} $\Gamma_F=(X,R,F)$,
  which does not admit any \textsc{subRAF}
  $\Gamma'_F=(X(R'),R',F\cap X(R'))$, but which admits multiple choice for
  the autocatalytic stoichiometric submatrix $S^*$.

We offer here two examples accordingly. The first reads
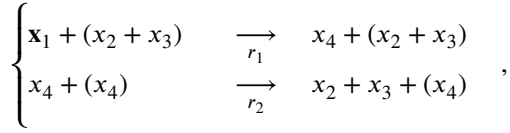
\[
    \begin{cases}
        \mathbf{x}_1 + (x_2 + x_3) \quad &\underset{r_1}\longrightarrow \quad x_4 + (x_2+x_3)\\
        x_4 + (x_4) \quad &\underset{r_2}\longrightarrow \quad x_2 + x_3 + (x_4)
    \end{cases}\quad,
\]
where we note that reaction $1$ has indeed two \textsc{F-catalysts}. Depending on this choice, we may identify two autocatalytic stoichiometric matrices $S^*_1$ and $S^*_2$, which are:
\begin{equation}\label{eq:ex2S1}
    \begin{split}
    S^*_1&=S[\{x_2,x_4\},R]=\begin{pmatrix}
       0 & 1\\
       1 & 0
    \end{pmatrix} \\[0.5em] \text{and} \quad S^*_2&=S[\{x_3,x_4\},R]=\begin{pmatrix}
       0 & 1\\
       1 & 0
        \end{pmatrix},
    \end{split}
\end{equation}
with associated circuits $\mathcal{C}$ as per T2 in Thm.~\ref{thm:main}:
\[
    \begin{split}
    \mathcal{C}_1&=r_1- x_2 - r_2 - x_4 - r_1\quad \\[0.5em]\text{and} 
    \quad \mathcal{C}_2&=r_1- x_3 - r_2 - x_4 - r_1.
    \end{split}
\]
Building on the second source of arbitrariness, we have the following example: 
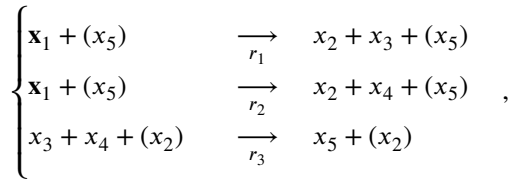
\[
\begin{cases}
  \mathbf{x}_1 + (x_5) \quad &\underset{r_1}\longrightarrow \quad x_2 + x_3+ (x_5)\\
  \mathbf{x}_1 + (x_5) \quad &\underset{r_2}\longrightarrow \quad x_2 + x_4 + (x_5)\\
x_3+x_4+(x_2)\quad &\underset{r_3}\longrightarrow \quad x_5 + (x_2)
\end{cases}\quad,
\]
where $x_2$ is net-product both of $r_1$ and $r_2$, both of \textsc{first generation}: $\gamma(r_1)=\gamma(r_2)=1$. In particular, again, we have two autocatalytic stoichiometric matrices $S^*_1$ and $S^*_2$, which are:
\begin{equation}\label{eq:ex2S2}
    \begin{split}
        S^*_1&=S[\{x_2,x_5\},\{r_1,r_3\}]=\begin{pmatrix}
       1 & 0\\
       0 & 1
    \end{pmatrix} \\[0.5em] 
    \text{and}  \quad S^*_2&=S[\{x_2,x_5\},\{r_2,r_3\}]=\begin{pmatrix}
       1 & 0\\
       0 & 1
    \end{pmatrix},     
    \end{split}   
\end{equation}
with associated circuits $\mathcal{C}$ as per T2 in Thm.~\ref{thm:main}:
\[
    \begin{split}
        \mathcal{C}_1&=r_1- x_5 - r_3 - x_2 - r_1\\[0.5em] 
        \text{and} \quad \mathcal{C}_2&=r_2- x_5 - r_3 - x_2 - r_2.    
    \end{split}
\]

In both examples \eqref{eq:ex2S1} and \eqref{eq:ex2S2}, the stoichiometric matrices $S^*_1$ and $S^*_2$ `differ' only in the associated species $X^*$ and $R^*$. We have deliberately made this choice to underline the symmetry of the construction and the underlying arbitrariness. Examples based on the same logic can be easily constructed to provide different matrices by changing the involved stoichiometric coefficients.
\end{eggsample}

\begin{eggsample}[Stoichiometric autocatalysis involving reactions with the
    same \textsc{F-reactants} and \textsc{F-products}.]
  \label{egg:7} In the framework adopted here, two reactions $r_1$ and $r_2$ with the same \textsc{F-reactants} and \textsc{F-products}, but different \textsc{F-catalysts}, are treated as distinct. In contrast, the original \textsc{RAF} framework considers such reactions as one, see Remark~\ref{rmk:eqclas}. As a consequence, the construction of the autocatalytic submatrix $S^*$ in Thm.~\ref{thm:main} may involve reactions that share the same \textsc{F-reactants} and \textsc{F-products}, as reactions $r_1$ and $r_2$ in the following example:
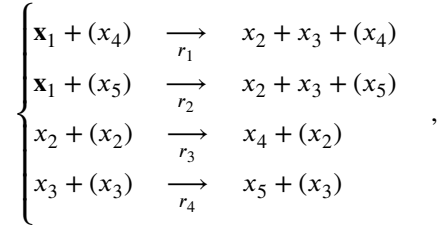
\[
    \begin{cases}
        \mathbf{x}_1+ (x_4) \quad \underset{r_1}\longrightarrow \quad x_2 + x_3 + (x_4)\\
         \mathbf{x}_1+ (x_5) \quad \underset{r_2}\longrightarrow\quad  x_2 + x_3 + (x_5)\\
        x_2 + (x_2) \quad \underset{r_3}\longrightarrow \quad x_4 + (x_2)\\
        x_3 + (x_3) \quad \underset{r_4}\longrightarrow \quad  x_5 + (x_3)
    \end{cases},
\]
The subset $X^*=\{x_2,x_3,x_4,x_5\}$ identifies an autocatalytic stoichiometric submatrix $S^*$
\[
S^*=S[X^*, R]=\begin{pmatrix}
    1 & 1 & -1 & 0\\
    1 & 1 & 0 & -1\\
    0 & 0 & 1 & 0\\
    0 & 0 & 0 & 1    
\end{pmatrix}
\]
with associated circuit $\mathcal{C}$
\[
\mathcal{C}= r_1- x_4- r_3 - x_2 - r_2 - x_5 - r_4 - x_3 - r_1,
\]
which contains both $r_1$ and $r_2$. This shows that stoichiometric autocatalysis may arise from the interplay between different catalyzed versions of the `same \textsc{F-reaction}' (i.e., with identical \textsc{F-reactants} and \textsc{F-products}). Nevertheless, the removal of either $r_1$ or $r_2$ necessarily identifies \textsc{subRAFs} $\Gamma'_F=(X, R\setminus \{r_1\}, F)$ and $\Gamma''_F=(X, R\setminus \{r_2\}, F)$, for which Thm.~\ref{thm:main} still applies and yields autocatalytic submatrices containing only one of $r_1$ or $r_2$, namely
\[
    \begin{split}
    S^*_1&=[\{x_2,x_4\},\{r_1,r_3\}]=\begin{pmatrix}
        1 & -1\\
        0 & 1
    \end{pmatrix}\\[0.5em]
    \text{and}\quad S^*_2&=[\{x_3,x_5\},\{r_2,r_4\}]=\begin{pmatrix}
        1 & -1\\
        0 & 1
    \end{pmatrix},
    \end{split}
\]
with associated circuits
\[
\mathcal{C}_1= r_1- x_4- r_3 - x_2 - r_1\quad \text{and}\quad \mathcal{C}_2=r_2 - x_5 - r_4 - x_3 - r_2.
\]
In conclusion, autocatalytic submatrices as in Thm.~\ref{thm:main} may involve reactions sharing the same \textsc{F-reactants} and \textsc{F-products}, but there also always exist such submatrices that do not involve multiple instances of these reactions.
\end{eggsample}

\begin{eggsample}[Not all stoichiometrically autocatalytic subnetworks
    possess properties T2-T3 in Thm.~\ref{thm:main}.]
  \label{egg:8} The statement of Thm.~\ref{thm:main} on the structural features of the autocatalytic stoichiometric submatrix $S^*$ follows specifically from the construction of a \textsc{RAF} that is not a \textsc{CAF}. Such features do not need to hold in general for stoichiometrically autocatalytic subnetworks.

Consider the following network:
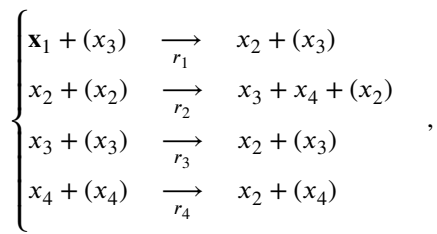
\[
    \begin{cases}
         \mathbf{x}_1 + (x_3) \quad\underset{r_1}\longrightarrow \quad  x_2 +(x_3)\\
        x_2  + (x_2) \quad \underset{r_2}\longrightarrow \quad x_3+x_4 + (x_2)\\
        x_3 + (x_3) \quad \underset{r_3}\longrightarrow \quad x_2 + (x_3)\\
        x_4 + (x_4) \quad \underset{r_4}\longrightarrow \quad x_2 + (x_4)
    \end{cases},
\]
which is easily verifiable to be a \textsc{RAF}, but not a \textsc{CAF}. Following the construction in Thm.~\ref{thm:main}, we identify the stoichiometrically autocatalytic submatrix 
\[
S^*=[\{x_2,x_3\},\{r_1,r_2\}]=\begin{pmatrix}
 1 & -1\\
 0 & 1
\end{pmatrix}
\]
with associated circuit $\mathcal{C}$:
\[
    r_1 - x_3 - r_2 -  x_2 - r_1.
\]
However, the stoichiometric submatrix:
\[
    S_C[\{x_2,x_3,x_4\},\{r_2,r_3,r_4\}]=\begin{pmatrix}
        -1 & 1 & 1\\
        1 & -1 & 0\\
        1 & 0 & -1
    \end{pmatrix},
\]
is as well autocatalytic. In the literature \cite{blokhuis_universal_2020}, such type of autocatalytic network is referred to as an \emph{autocatalytic core of type III}. Conditions T2-T3 stated in Thm.~\ref{thm:main} do not hold for $S_C$: there is no elementary circuit where all species $\{x_2,x_3,x_4\}$ and reactions $\{r_2,r_3,r_4\}$ belong, and all columns of $S$ possess one negative entry. 

Finally, via Remark~\ref{rmk:thm1}, we repeat that we may as well consider stoichiometrically autocatalytic \textsc{RAFs} where all species are in the food set. However, since such networks are always \textsc{CAFs}, it is natural that the conditions T1-T2-T3 do not necessarily apply.
\end{eggsample}

\section{Discussion}

In this paper, we showed the intimate relationship between two frameworks
of autocatalysis that superficially appear distinct. The \textsc{RAF}
framework is ignorant of any quantitative stoichiometric consideration,
while the definition of \emph{stoichiometric autocatalysis} is grounded in
the matrix concept of semipositivity, which requires quantitative
specification of the stoichiometric coefficients.

The main result, Thm.~\ref{thm:main}, states that the stoichiometric matrix
of any \textsc{RAF} that does not fall into the special case of being
itself a \textsc{CAF} necessarily contains at least one autocatalytic
submatrix. This defines such a RAF, viewed as a reaction network, as
stoichiometrically autocatalytic. We achieved this result by explicitly
constructing a square autocatalytic submatrix of the stoichiometric matrix
using only reactant/product relations and -- consistently with the RAF
formulation -- without requiring any explicit specification of
stoichiometric coefficients. In particular, such a construction is possible
whenever all reactions possess at least one reactant that is not in the
food set. Since any \textsc{RAF} that is not a \textsc{CAF} must contain a
\textsc{subRAF} (with respect to a suitably enlarged food set) with this
property, the result follows.

We also note that the constructed matrix is -- up to reordering of the
columns -- an instance of a so-called Child-Selection matrix
\cite{vassena_unstable_2024}, where each species is bijectively associated
with a reaction in which it appears as a reactant. When each reaction has at least one catalyst not in the food set, this reactant can be chosen to be one of
its catalysts. The relevance of Child-Selection matrices in the analysis of
dynamics, stability, and bifurcations of steady states has been discussed
in the literature \cite{vassena_unstable_2024,
  blokhuis_stoichiometric_recipes}.

The examples in Sec.~\ref{sect:examples} further illustrate the subtlety of
the relationship between these two frameworks. In particular, they show
that the \textsc{RAF} structure cannot, in general, be inferred solely from
the specification of the food set (see Examples in~\ref{sect:examples}.\ref{egg:ex1} and \ref{sect:examples}.\ref{egg:ex2}). Such an inference could, however, always be achieved by
enforcing additional structural requirements; for instance, one may
identify \textsc{F}-reactants with net reactants and \textsc{F}-catalysts
with catalysts, an assumption that appears consistent with the original
biological motivation.

At the same time, the identification of \textsc{RAF}s through
stoichiometric autocatalysis remains delicate. As illustrated
in \ref{sect:examples}.\ref{egg:5} and \ref{sect:examples}.\ref{egg:6}, the correspondence between the two objects
is neither injective nor surjective: multiple distinct \textsc{RAFs} may
correspond to a single autocatalytic stoichiometric submatrix, while a
single \textsc{RAF} may admit several independent autocatalytic
submatrices.

Finally, the example in \ref{sect:examples}.\ref{egg:8} highlights that the
structure of the autocatalytic submatrix described in Thm.~\ref{thm:main}
is specific to the \textsc{RAF} setting and need not be present, in
general, in stoichiometric autocatalysis per se. A more systematic
investigation of this specific structure appears as a natural direction for
future work.

\bigskip

\textbf{Acknowledgments:} This work has been supported by the Novo Nordisk
Foundation (grant NNF21OC0066551 `MATOMIC'), during a visit of Wim Hordijk
in Leipzig.  Research in the Stadler lab is support by the BMBF (Germany)
through DAAD project 57616814 (SECAI, School of Embedded Composite AI), and
jointly with SMWK (Saxony) through the \emph{Center for Scalable Data
  Analytics and Artificial Intelligence Dresden/Leipzig} (SCADS24B).

\bibliographystyle{cas-model2-names}
\bibliography{references.bib}

\end{document}